\documentclass[pra,twocolumn,nofootinbib]{revtex4-1}

\usepackage{amssymb,amsfonts,amsmath,amsthm}
\usepackage{graphicx}

\newcommand{\ket}[1]{ |{#1}\rangle}
\newcommand{\bra}[1]{ \langle{#1}|}

\newcommand{\ketbra}[2]{| {#1} \rangle\!\langle {#2} |}
\newcommand{\proj}[1]{| {#1} \rangle\!\langle {#1} |}

\newcommand{\Tr}[0]{{\rm Tr}}
\newcommand{\Hil}{\mathcal{H}}

\newcommand{\Id}{\mathbb{I}}
\newcommand{\Lin}{\mathcal{L}}

\newcommand{\JL}{L}

\newcommand{\Jamiolkowski}{Jamio\l kowski }

\newtheorem{theorem}{Theorem}
\newtheorem{definition}[theorem]{Definition}
\newtheorem{lemma}[theorem]{Lemma} 
\newtheorem{corollary}[theorem]{Corollary}

\begin{document}

\title{Implementing positive maps with multiple copies of an input state}

\author{Qingxiuxiong Dong}
 \email{dong@eve.phys.s.u-tokyo.ac.jp}
 \affiliation{Department of Physics, Graduate School of Science, The University of Tokyo, Tokyo, Japan}
\author{Marco T{\'u}lio Quintino}
 \email{quintino@eve.phys.s.u-tokyo.ac.jp}
 \affiliation{Department of Physics, Graduate School of Science, The University of Tokyo, Tokyo, Japan}
\author{Akihito Soeda}
 \email{soeda@phys.s.u-tokyo.ac.jp}
 \affiliation{Department of Physics, Graduate School of Science, The University of Tokyo, Tokyo, Japan}
\author{Mio Murao}
 \email{murao@phys.s.u-tokyo.ac.jp}
 \affiliation{Department of Physics, Graduate School of Science, The University of Tokyo, Tokyo, Japan}

\date{\today}

\begin{abstract}
Valid transformations between quantum states are necessarily described by {\it completely} positive maps, instead of just positive maps.  
Positive but not completely positive maps  such as the transposition map cannot be implemented due to the existence of entanglement in composite quantum systems,  
but there are classes of states for which the positivity is guaranteed, e.g., states not correlated to other systems.
In this paper, we introduce the concept of {\it $N$-copy extension} of maps  to {\it quantitatively} analyze the difference between positive maps and completely positive maps.  
We consider implementations of the action of positive  but not completely positive maps on uncorrelated states by allowing an extra resource of consuming multiple copies of the input state and characterize the positive maps in terms of implementability with multiple copies.   
We show that by consuming multiple copies, the set of implementable positive maps becomes larger, and almost all positive maps are implementable with finite copies of an input state.  
The number of copies of the input state required to implement a positive map quantifies the degree by which a positive map violates complete positivity.
We then analyze the optimal $N$-copy implementability of a noisy version of the transposition map. 

\end{abstract}

\maketitle

\section{introduction}\label{sec:introduction}

Quantum information processing is performed by transformations between quantum states,
and thus which class of transformations are allowed in quantum mechanics 
is a fundamental  issue in quantum information processing. 
In quantum mechanics, valid deterministic transformations are called quantum channels, 
which are  mathematically described as completely positive (CP) and trace preserving (TP) maps~\cite{nielsenchuang,kraus}.
The TP condition can be relaxed by allowing a probability of failure.
 Non-TP maps have been investigated for understanding probabilistic properties of quantum mechanics,
and play a key role in protocols such as quantum teleportation~\cite{teleportation}.
In contrast, for the CP condition of the CPTP maps, it has not been much considered if we can implement non-CP maps by relaxing certain requirement.

In the density operator formalism of quantum mechanics, a general quantum state is described by a density operator, namely, a positive operator with its trace one.   
Valid transformations between states have to preserve the positivity of quantum states, which leads to the positivity requirement of maps.  
Moreover, when a composite system is considered,  the linearity of quantum mechanics extends the action of a map to the composite system.  
Complete positivity of a map is equal to the positivity of that map on a part of a composite system, and it is required for a map to be a valid transformation in quantum mechanics.  
The positivity of a map does not promise its complete positivity, and there exist positive but not completely positive (PNCP) maps. 
 The transposition map is an example of PNCP maps, which plays an important role in entanglement criterion~\cite{ppt1,ppt2}.

\begin{figure}[t]
\includegraphics[width=\hsize]{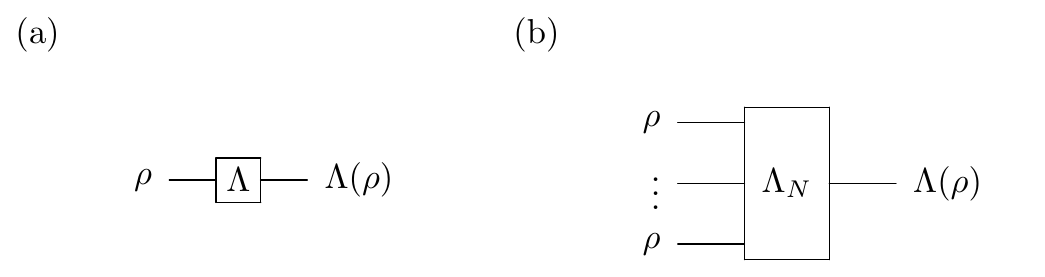}
\caption{
 A positive but not completely positive (PNCP) map $\Lambda$ is not implementable by  the quantum circuit with a single input state $\rho$ represented by (a), but it may be implementable by  another type of the quantum circuit corresponding to $\Lambda_N$ with $N$ copies of an input state represented by (b).   If the quantum circuit of type (b) exists, we say that $\Lambda$ is $N$-copy implementable.}\label{fig:circuit}

\end{figure}

Even if we promise that the input state is not correlated to other quantum systems,
the linearity of quantum mechanics defines the actions for correlated states.
As the linearity of quantum mechanics is one reason that PNCP maps are not implementable,
we consider one relaxation to this condition, that is, adding multiple copies of an input state.   As we will show in this paper, the assumption of adding multiple copies reveals fundamental properties of positive maps.
It also contributes to practical scenarios  where multiple copies of an input state are initially available.
These copies serve as an extra resource as the no-cloning theorem~\cite{no_cloning} forbids producing perfect copies from a single input state.

 A trivial example that multiple copies of an input state help in implementing PNCP maps is the case when infinite copies of an input state are available.   In this case, we can perform state tomography~\cite{tomography_vogel,tomography_paris} on the input state $\rho$, and obtain a classical description of $\rho$.
We then  (classically) calculate the action of $\Lambda$ on $\rho$ using its classical description, 
which results in a classical description of $\Lambda(\rho)$.
Since $\rho$ is a positive operator and $\Lambda$ is a positive map, $\Lambda(\rho)$ is also a positive operator and  we can prepare the quantum state $\Lambda(\rho)$ with normalization.  By employing this  ``measure, estimate, calculate and prepare strategy", any positive map can be implemented for input states not correlated to other systems if infinite copies of the input state  are provided.

However, if only finite copies of an input state are available, how does the implementability of PNCP maps change?
It is shown that if only finite copies are available, the transposition map is not implementable in an exact manner, even probabilistically\footnote{Although Ref.~\cite{certifying} does not state this result explicitly, it follows from the section 6 of the appendix that the optimal approximation of this probabilistic transformation cannot have average fidelity equals one, hence it can never be exact. }~\cite{certifying,miyazaki}.
That is, in order to implement the transposition map, infinite copies of an input state are necessary.
On the other hand, this result does not exclude the possibility of implementing other PNCP maps with finite copies,
and it is not known which class of maps become implementable when finite copies are available.

For the finite copy case,  many of preceding studies focus on evaluating optimal approximate implementations of  maps, such as universal NOT~\cite{universal_not} and quantum cloning~\cite{cloning_gisin,cloning_werner}.    In contrast, we focus on {\it exact} implementations  in this paper, as the evaluation of approximation is not unique 
 (e.g., average fidelity, diamond norm~\cite{ksv,buscemi03,watrous}, or structural physical approximation~\cite{horodecki01,psa}).
The requirement of  the exact  implementation indicates that the protocols  in the type of measure, estimate, calculate and prepare strategy does not help, as the result is always approximate.

We introduce the concept of {\it $N$-copy extension} of a map.
We analyze the implementability of the action of a PNCP map $\Lambda$ 
when $N$ copies of an input state $\rho$,
which is not correlated to other systems, are given.
Mathematically, we consider the existence of {\it CP $N$-copy extension of map}, a CP map $\Lambda_N$ 
which maps $N$ copies of $\rho$, that is $\rho^{\otimes N}$, to a single state $\Lambda(\rho)$,
namely, $\Lambda_N(\rho^{\otimes N}) = \Lambda(\rho)$ for every $\rho$.
 We figure out how the set of implementable maps changes when the number of copies, $N$, changes.

The problem of CP $N$-copy extension can then be regarded as a kind of the CP extension problem,
which has been studied for two-dimensional  systems~\cite{extension_twodim}, pure states in general dimensional systems with the Gram matrices~\cite{extension_pure}.
For general states with an arbitrary dimension, this problem can be tackled by semidefinite programming (SDP) methods~\cite{extending_qmaps}.
In this work we do not make use of SDP methods. 
Instead, we show that the CP $N$-copy extension problem can be solved by evaluating the smallest eigenvalue of an operator,  which is considerably simpler than solving a general SDP problem or the SDP formulation presented in
Ref.~\cite{extending_qmaps}.

We show that almost all PNCP maps become implementable if enough  but still finite copies of an input state are provided.   The number of copies required for implementing a positive map can be regarded as a measure of how difficult it is to implement the map.

Positive maps have a one-to-one correspondence to entanglement witness~\cite{ppt2},
and characterizations of positive maps based on entanglement witness have been studied in~\cite{separability}.
Our results on $N$-copy implementable maps give quantitative bound to the hierarchies in~\cite{separability}.
Structural physical approximation~\cite{horodecki01,psa} provides a method for quantifying the difficulty for implementing positive maps restricted to a single copy of the input state.
The idea of $N$-copy implementability presented in this paper provides a different quantification for the difficulty,
the number of copies of the input state which is required to implement a positive map.

This paper is organized as follows.  
In Section~2, we introduce the concept of $N$-copy extension.
In Section~3, we give a universal construction of $N$-copy extension. 
This construction is shown to be an optimal construction,
and $N$-copy implementability of a map is verifiable by only considering this construction.
While this result fully characterizes $N$-copy implementability, 
we give two sufficient and one necessary conditions for $N$-copy implementability in Section~4 and 5, 
which are simpler to verify.
In Section~4, we consider two noisy versions of a positive map,
and we give sufficient conditions for the noisy versions to be $N$-copy implementable.
The result also shows that all non-boundary positive maps are implementable with finite copies of an input state.
In Section~5, a necessary condition for a positive map to be $N$-copy implementable is given.
In Section~6, we consider the boundary and extremal conditions of a positive map with a few examples.
In Section~7, we analyze the $N$-copy implementability of the noisy version of the transposition map in general dimension.
In Section~8, we discuss a characterization of positive maps based on $N$-copy implementability.

%%%%%%%%%%%%%
\section{Implementing maps with multiple copies of an input state}

Let $\mathcal{D}(\Hil)$ be the set of density operators on the Hilbert space $\Hil$.
A linear map $\Lambda : \Lin(\Hil_1) \to \Lin(\Hil_0)$ is said to be implementable in quantum mechanics 
if there exists a quantum circuit such that it (probabilistically) transforms any input state $\rho \in \mathcal{D}(\Hil_1)$ 
into $\Lambda(\rho) \in \mathcal{D}(\Hil_0)$ as shown in Fig.~\ref{fig:circuit} (a).   
Such a quantum circuit exists if $\Lambda$ is a CP map.

We consider a situation where $N$ copies of an input state $\rho$ are provided to  perform a task to prepare a single state given by $\Lambda(\rho)$.    We say that $\Lambda$ is implementable with $N$ copies of an input, or simply \emph{$N$-copy implementable},
when there exists a quantum circuit that transforms any input state $\rho^{\otimes N}$ into $\Lambda(\rho)$ as shown in Fig.~\ref{fig:circuit}(b).

The quantum circuit  represented by Fig.~\ref{fig:circuit}(b) 
 satisfies the following properties.  The input quantum system of the quantum circuit is not just $\Hil_1$ 
but $N$ copies of $\Hil_1$.   
In order to distinguish all input quantum systems, we define  $\Hil_i \cong \mathbb{C}^{d_1}$ for $i = 2,\ldots,N$, where $d_1 = \dim \Hil_1$.
The subscripts of vectors and operators also denote the Hilbert spaces in or on which they are defined, e.g., $\ket{\phi}_j \ket{\psi}_i = \ket{\psi}_i \ket{\phi}_j$, and the subscripts are omitted if trivial from the context.
The total input quantum system is $\Hil_1 \otimes \cdots \otimes \Hil_N$,
and the extended map is a linear map from $\Lin(\Hil_1 \otimes \Hil_2 \otimes \cdots \otimes \Hil_N)$ to $\Lin(\Hil_0)$,
which is denoted by $\Lambda_N$.
We denote the Hilbert space $\Hil_1 \otimes \Hil_2 \otimes \cdots \otimes \Hil_N$ as $\Hil_{12\cdots N}$ for convenience.
Our aim is to find an implementation of a map that transforms a state 
$\rho$ to $\Lambda(\rho)$ with $N$ copies of $\rho$.
Thus, a requirement for $\Lambda_N$ is that 
$\Lambda_N$ maps $\rho^{\otimes N} \in \Lin(\Hil_{12\cdots N})$ to $\Lambda(\rho)$.
If this condition is satisfied, we call the map $\Lambda_N$ as an {\it $N$-copy extension of $\Lambda$},
and if $\Lambda_N$ is CP, we call it as a {\it CP $N$-copy extension of $\Lambda$}.
The $N$-copy implementability of $\Lambda$ is  equivalent  to the existence of a CP $N$-copy extension of $\Lambda$.

\begin{definition}[$N$-copy extension]\label{def:n_copy_extension}
Given a linear map 
$
\Lambda :\Lin(\Hil_1) \to  \Lin(\Hil_0),
$
a linear map
$
\Lambda_N : \Lin(\Hil_1 \otimes \Hil_2 \otimes \cdots \otimes \Hil_N) \to \Lin(\Hil_0)
$
is an $N$-copy extension of $\Lambda$ if $\Lambda_N$ satisfies 
\begin{align}
\Lambda_N (\rho^{\otimes N}) = \Lambda( \rho ) \label{eq:def_n_copy_extension}
\end{align}
 for any quantum state $\rho \in \mathcal{D}(\mathbb{C}^{d_1})$.
 When the $N$-copy extension is CP,
 it is a CP $N$-copy extension.
\end{definition}

Due to the linearity, the action of $\Lambda_N$ is not only determined on $\rho^{\otimes N}$, 
but on the linear span of $\rho^{\otimes N}$,
where the linear span is defined as 
\begin{align}
\mathrm{span}\{ \rho^{\otimes N} \} := 
\left\{ \sum_i c_i \rho_i^{\otimes N} \middle| c_i \in \mathbb{C}, \rho_i \in \mathcal{D}(\mathbb{C}^{d_1}) \right\}. \label{eq:span_rho}
\end{align}
An element in the linear span $ {O}_N \in \mathrm{span}\{ \rho^{\otimes N} \}$ can be decomposed as
$
{O}_N = \sum_i c_i \rho_i^{\otimes N},
$
and the action of $\Lambda_N$ on ${O}_N$  can be evaluated with linearity and Eq.~\eqref{eq:def_n_copy_extension} as
\begin{align}
\Lambda_N ( {O}_N ) &= 
 \sum_i c_i \Lambda_N( \rho_i^{\otimes N} ) 
= \sum_i c_i \Lambda( \rho_i) \notag \\
&= \Lambda (\sum_i c_i \rho_i) 
= \Lambda ( \Tr_{2,\ldots,N} {O}_N ). \label{eq:calc_by_partial_trace}
\end{align}

Note that for an operator $a \in \Lin(\mathbb{C}^{d_1})$ which is not a density operator,
even if $a^{\otimes N} \in \mathrm{span}\{ \rho^{\otimes N} \}$, 
the condition $\Lambda_N( a^{\otimes N} ) = \Lambda(a)$ is not guaranteed.
For example, consider $a =\ketbra{0}{1}$,
the action of $\Lambda_N$ on $\ketbra{0}{1}^{\otimes N}$ is not $\Lambda(\ketbra{0}{1})$ but $0$,
as the partial trace over $\Hil_2,\ldots,\Hil_N$ of $\ketbra{0}{1}^{\otimes N}$ is $0$, namely,
\begin{align}
\Lambda_N ( a^{\otimes N} ) = \Lambda ( \Tr_{2,\ldots,N}\, a^{\otimes N} ) = \Lambda (0) = 0.
\end{align}

The problem of CP N-copy extension is a particular case of the CP extension problem, 
the objective of which is to decide for a given map defined on a subspace whether there exists a CP map defined on the full space preserving the action on the subspace. 
In particular, for a given set of input states $\{ \rho_i\}$ and a set of output states $\{ \sigma_i \}$, the CP extension problem is to verify the existence of a CP map or sometimes a quantum channel (a CPTP map)  that maps $\rho_i$ into $\sigma_i$ for every $i$.   
For the CP $N$-copy extension problem, 
the set of input states is $\{ \rho^{\otimes N} \}$ and the set of output states is $\{ \Lambda(\rho) \}$.    
The CP extension problem has been studied for two-dimensional  systems in~\cite{extension_twodim} and for pure states in general  dimensional systems in~\cite{extension_pure} using  the Gram matrices.   
A recent study for general states in general dimensional systems~\cite{extending_qmaps} shows that the CP extension problem can be formulated as an SDP.
 The problem of CP $N$-copy extension is a subclass of this general extension problem.

%%%%%%%%%%%%%
\section{$N$-copy Implementable maps}\label{sec:n_copy_implementability}

We present a characterization of $N$-copy CP extensible maps based on the positivity of an operator in this section.
That is, we show an universal construction of $\Lambda_N$,  and the $N$-copy implementability of $\Lambda$ can be verified by analyzing this construction.  
This is achieved by noticing that the subspace in our problem satisfies permutation invariance under permutations of the input Hilbert spaces.

\begin{theorem}\label{thm:optimality_of_sym_ext}
There exists a CP $N$-copy extension of $\Lambda$ if and only if
the $N$-copy extension $\Lambda^{sym}_N$ is CP,
where
\begin{align}
\Lambda^{sym}_N ( \rho_1 \otimes \rho_2 \otimes \cdots \otimes \rho_N ) 
:= \frac{1}{N} \sum_{i=1}^N \Lambda(\rho_i) \prod_{j \neq i} \Tr_j \rho_j ,\label{eq:n_copy_extension}
\end{align}
for all $\rho_i \in \Lin(\Hil_i)$.
Moreover, if $\Lambda$ is TP, $\Lambda_N^{sym}$ is also TP.
\end{theorem}

Note that $\Lambda^{sym}_N$ is an $N$-copy extension of $\Lambda$, 
that is, if we set $\rho_i = \rho $ for $i = 1,\ldots,N$, we obtain
\begin{align}
\Lambda^{sym}_N ( \rho^{\otimes N} ) 
= \frac{1}{N} \sum_{i=1}^N \Lambda(\rho) = \Lambda(\rho).
\end{align}

In order to verify if a linear map is CP, we utilize the Choi-\Jamiolkowski isomorphism~\cite{positive_maps}.
This isomorphism maps a linear map $\Lambda : \Lin(\Hil_1) \to \Lin(\Hil_0)$ to an operator $L \in \Lin(\Hil_1 \otimes \Hil_0)$
as $L := \sum \ketbra{i}{j} \otimes \Lambda(\ketbra{i}{j})$.
$L$ is called as the Choi operator of $\Lambda$.
The definition of $L$ is also equivalent to $L = id \otimes \Lambda( d_1 \proj{\Phi^+} )$ where 
$\ket{\Phi^+} := (1/ \sqrt{d_1}) \sum_{i=0}^{d_1 - 1} \ket{i}\ket{i}$ is the maximally entangled state on $\Hil_1 \otimes \Hil_1$.
Since  complete positivity of a map is equivalent to the positivity of the corresponding Choi operator,
we can verify  complete positivity of a map by calculating the smallest eigenvalue of its Choi operator.

Theorem~\ref{thm:optimality_of_sym_ext} states that for a given positive map $\Lambda$,
its $N$-copy implementability is equal to the  complete positivity of $\Lambda^{sym}_N$.
Let $L$ be the Choi operator of $\Lambda$, 
the  complete positivity of $\Lambda^{sym}_N$ is equivalent to the positivity of the corresponding Choi operator 
$\JL^{sym}_N \in \Lin(\Hil_{01\cdots N})$ given by
\begin{align}
\JL^{sym}_N = \frac{1}{N} \sum_{i=1}^N (L)_{0i} \otimes \Id_{R}, \label{eq:n_copy_extension_choi}
\end{align}
where the  subscripts $0,i$ on the r.h.s.~denotes the Hilbert space $\Hil_0,\Hil_i$ and $R$ denotes the rest.
Thus, Theorem~\ref{thm:optimality_of_sym_ext} is equivalent to that 
a positive map $\Lambda$ is $N$-copy implementable if and only if 
its corresponding $\JL^{sym}_N$ defined as Eq.~\eqref{eq:n_copy_extension_choi} is positive.
Various properties of the operator which has a symmetrized structure as in Eq.~\eqref{eq:n_copy_extension_choi} have been studied~\cite{separability,math_entanglement,lancien2016k}.
In this paper, we focus on the non-negativity of the minimum eigenvalue of the operator in Eq.~\eqref{eq:n_copy_extension_choi}, which corresponds to the existence of a CP $N$-copy extension.

While a characterization of $N$-copy implementability can be formulated as an SDP~\cite{extending_qmaps},
we show that $N$-copy implementability of a map $\Lambda$ is decidable by calculating the smallest eigenvalue of the operator given by Eq.~\eqref{eq:n_copy_extension_choi},
which is considerably simpler than solving an SDP.
If the $N$-copy implementability problem is formulated as an SDP,  our construction of $\Lambda_N$ always corresponds to the candidate of the optimal solution of the SDP,   
and thus we only need to see if this construction satisfies the SDP condition to determine if there exists a solution to this SDP.

Theorem~\ref{thm:optimality_of_sym_ext} can also be used to show the positivity of a class of maps.
Since any $N$-copy implementable map is necessarily a positive map, 
we can verify the positivity of a map if it is already $N$-copy implementable.
A characterization of positive maps in terms of $N$-copy implementability is discussed in Section~\ref{sec:characterization}.

\begin{proof}[Proof of Theorem~\ref{thm:optimality_of_sym_ext}]
The ``if part'' is trivial, and we show the ``only if part''.
Let $\Lambda_N$ be an $N$-copy extension of $\Lambda$,
and the Choi operators of $\Lambda_N$ and $\Lambda^{sym}_N$ be 
$\JL_N$ and $\JL^{sym}_N$, respectively.
Since $\Lambda_N (\Lambda_N^{sym})$ is CP if and only if $\JL_N (\JL_N^{sym}) $ is positive, 
it is enough to show that $\JL^{sym}_N \geq 0$ holds when $\JL_N \geq 0$.

Consider the permutations $\pi$ of $N$ input states, namely, $\Lambda_N \circ \pi (\rho_1 \otimes \cdots \otimes \rho_N) 
= \Lambda_N (\rho_{\pi(1)} \otimes \cdots \otimes \rho_{\pi(N)})$,
and the average of them $\mathrm{Sym}(\Lambda_N) = \frac{1}{N!} \sum_{\pi \in S_N} \Lambda_N \circ \pi$.
For any $\rho_i$, we have
\begin{align}
&\frac{1}{N!} \sum_{\pi \in S_N} \Lambda_N \circ \pi (\rho_1 \otimes \cdots \otimes \rho_N)  \\
&= \Lambda_N ( \frac{1}{N!} \sum_{\pi \in S_N} \rho_{\pi(1)} \otimes \cdots \otimes \rho_{\pi(N)}).
\end{align}
Here,
$\frac{1}{N!} \sum_{\pi \in S_N} \rho_{\pi(1)} \otimes \cdots \otimes \rho_{\pi(N)} \in \mathrm{span}\{ \rho^{\otimes N}\}$,
which can be derived by a similar argument to the proof of Theorem 3 in~\cite{church_of_sym} as follows.
Let $R( c_1, c_2, \ldots, c_N ) = \sum_{i=1}^N c_i \rho_i$ with $c_i \in \mathbb{R}$ for $i = 1,\ldots,N$.
Then $R( c_1, c_2, \ldots, c_N )^{\otimes N} \in \mathrm{span}\{ \rho^{\otimes N}\}$ holds by definition.
 Notice that $\sum_{\pi \in S_N} \rho_{\pi(1)} \otimes \cdots \otimes \rho_{\pi(N)} 
= \frac{\partial^N}{\partial c_1 \partial c_2 \cdots \partial c_N} R( c_1, c_2, \ldots, c_N )^{\otimes N}$,
proving that $\frac{1}{N!} \sum_{\pi \in S_N} \rho_{\pi(1)} \otimes \cdots \otimes \rho_{\pi(N)} \in \mathrm{span}\{ \rho^{\otimes N}\}$.
Thus, the action of $\mathrm{Sym}(\Lambda_N) $ on $\rho_1 \otimes \cdots \otimes \rho_N$ can be calculated invoking Eq.~\eqref{eq:calc_by_partial_trace}, and is uniquely determined by Eq.~\eqref{eq:def_n_copy_extension}, 
i.e., $\mathrm{Sym}(\Lambda_N) = \mathrm{Sym}(\Lambda'_N)$ for any $\Lambda'_N$ satisfying Eq.~\eqref{eq:def_n_copy_extension}.
Since $\Lambda^{sym}_N$ satisfies Eq.~\eqref{eq:def_n_copy_extension}, 
$\mathrm{Sym}(\Lambda_N) $ can be calculated by taking $\Lambda^{sym}_N$ in place of $\Lambda_N$, 
 and we obtain $\mathrm{Sym}(\Lambda_N) = \frac{1}{N!} \sum_{\pi \in S_N} \Lambda^{sym}_N \circ \pi = \Lambda^{sym}_N$. 
 In other words, 
$
 \frac{1}{N!} \sum_{\pi \in S_N} \Lambda_N \circ \pi = \Lambda^{sym}_N \label{eq:lambda_sym_from_lambda_n}
$
 holds for any $\Lambda_N$ satisfying Eq.~\eqref{eq:def_n_copy_extension}.
The Choi operator of $\Lambda_N\circ\pi$ is given by 
$( id \otimes \Lambda_N \circ \pi) ( d_1^N \proj{\Phi^+}) = (\pi^{-1} \otimes \Lambda_N) (d_1^N \proj{\Phi^+}) 
= (\pi^{-1}\otimes id) (\JL_N)$,
and the corresponding eigenvalues are the same as the eigenvalues of $\Lambda_N$ for any $\pi$.
Thus, the smallest eigenvalue of the Choi operator of $\Lambda^{sym}_N$ is bounded as
\begin{align}
\lambda_{\min} (\JL^{sym}_N) &= \lambda_{\min} \left[ \frac{1}{N!}  \sum_{\pi \in S_N} (\pi^{-1}\otimes id) (\JL_N) \right] \\
&\geq \lambda_{\min} ( \JL_N ) .
\end{align}
Therefore, for $\JL_N \geq 0$, we obtain $\JL^{sym}_N \geq 0$.
\end{proof}

In the  rest of this paper except for Section~\ref{sec:no_n_copy_extension}, 
we use $\Lambda_N$ to denote the specific $N$-copy extension of $\Lambda$ given by Eq.~\eqref{eq:n_copy_extension},
and $\JL_N$ the corresponding Choi operator given by Eq.~\eqref{eq:n_copy_extension_choi}.

\section{Noisy version of positive maps}\label{sec:implement_n_copy_extension}

While Theorem~\ref{thm:optimality_of_sym_ext} gives a method to verify the $N$-copy implementability of a map,
it does not show how the set of $N$-copy implementable maps changes depending on $N$.
One way to capture the change of the set of implementable maps in $N$ 
is to analyze a set of maps written in a certain form.
For example, the implementability of a noisy version of two dimensional state transposition defined as
\begin{align}
T^\eta (\rho) = (1-\eta) \rho^T + \eta \Id / 2
\end{align}
can be  obtained by using the results of~\cite{universal_not}.
 When $N$ copies of an input state are available, this  noisy version of the map is implementable if and only if $\eta \geq 2/(N+2)$.

In this section, we consider the implementability of a wider class of maps given by 
\begin{align}
\Lambda^{\eta_a} (\rho) := (1-\eta_a) \Lambda(\rho) + \eta_a \frac{\Tr\JL}{d_1}  \frac{\Id_0}{d_0} \Tr(\rho), \label{eq:white_noise_robustness}
\end{align}
where $\Lambda :\Lin(\Hil_1) \to  \Lin(\Hil_0)$ is an arbitrary positive map,
and $\JL$ is the Choi operator of $\Lambda$. 
 The parameter $\eta_a$ may be seen as the amount of white noise that is added to the original map $\Lambda$.
When $N$ copies of an input state are available, we give a general bound on $\eta_a$ independent of $\Lambda$ 
(Theorem~\ref{thm:n_copy_extension_add_white_noise}).
Note that while $\Lambda^{\eta_a}$ can be regarded as an approximation of $\Lambda$,
what we show is the exact implementability of $\Lambda^{\eta_a}$.

We can also consider another kind of a noisy version of the map by introducing a decomposition given by 
\begin{align}
\Lambda^{\eta_b} (\rho) := (1-\eta_b) \Lambda(\rho) + \eta_b \Lambda( \frac{\Id_1}{d_1} \Tr(\rho) ). \label{eq:depolarizing_first}
\end{align}
In Lemma~\ref{lem:navascues_ext}, another bound is derived for this decomposition.
While Eq.~\eqref{eq:white_noise_robustness} is the convex sum of $\Lambda$ 
and the  renormalized totally depolarizing channel,
Eq.~\eqref{eq:depolarizing_first} can be interpreted as a map that first depolarize the state, 
and then apply $\Lambda $ on the depolarized state, 
that is, $\Lambda^{\eta_b} = \Lambda \circ ((1-\eta_b) id + \eta_b  \frac{\Id}{d_1} \Tr)$.
We also give a general bound on $\eta_b$ independent of $\Lambda$ in Theorem~\ref{thm:implementable_eta}.
Note that Eq.~\eqref{eq:depolarizing_first}  coincides with Eq.~\eqref{eq:white_noise_robustness} when $\Lambda$ is unital,
and in this case, Theorem~\ref{thm:implementable_eta} gives a better bound.

\begin{theorem}\label{thm:n_copy_extension_add_white_noise}
For any positive map 
$
\Lambda :\Lin(\Hil_1) \to  \Lin(\Hil_0),
$
its noisy version $\Lambda^{\eta_a} (\rho)$ defined as Eq.~\eqref{eq:white_noise_robustness} is $N$-copy implementable
if $\eta_a \geq d_0 d_1^2 / (N + d_0 d_1^2) $, where $d_i = \dim \Hil_i$.
\end{theorem}

\begin{theorem}\label{thm:implementable_eta}
For any positive map 
$
\Lambda :\Lin(\Hil_1) \to  \Lin(\Hil_0),
$
its noisy version $\Lambda^{\eta_b} (\rho)$ defined as Eq.~\eqref{eq:depolarizing_first} is $N$-copy implementable
if $\eta_b \geq {d_1^2}/({N+d_1^2}) $, where $d_i = \dim \Hil_i$.
\end{theorem}

For any positive map $\Lambda$, 
Theorem~\ref{thm:n_copy_extension_add_white_noise} states that by adding white noise of an amount of $O(1/N)$,
the noisy version of $\Lambda$ always becomes $N$-copy implementable.
While this result does not characterize the whole set of $N$-copy implementable maps,
it shows a subset of $N$-copy implementable maps (Fig.~\ref{fig:set_representation}(b)).

Theorem~\ref{thm:n_copy_extension_add_white_noise} also indicates that for any $\eta_a>0$, 
there exists a number $N$ such that $\Lambda^{\eta_a}$ is $N$-copy implementable,  that is, finite-copy implementable.
Since the set of positive maps is convex, any non-boundary (interior) positive map
can be decomposed as a convex sum of another positive map and the completely depolarizing channel.
Thus, we obtain the following corollary.
\begin{corollary}\label{cor:non_boundary_implementable}
Any non-boundary positive map is finite-copy implementable.
\end{corollary}

\begin{figure}[t]
\centering
\flushleft{(a)}
\includegraphics[width=85mm]{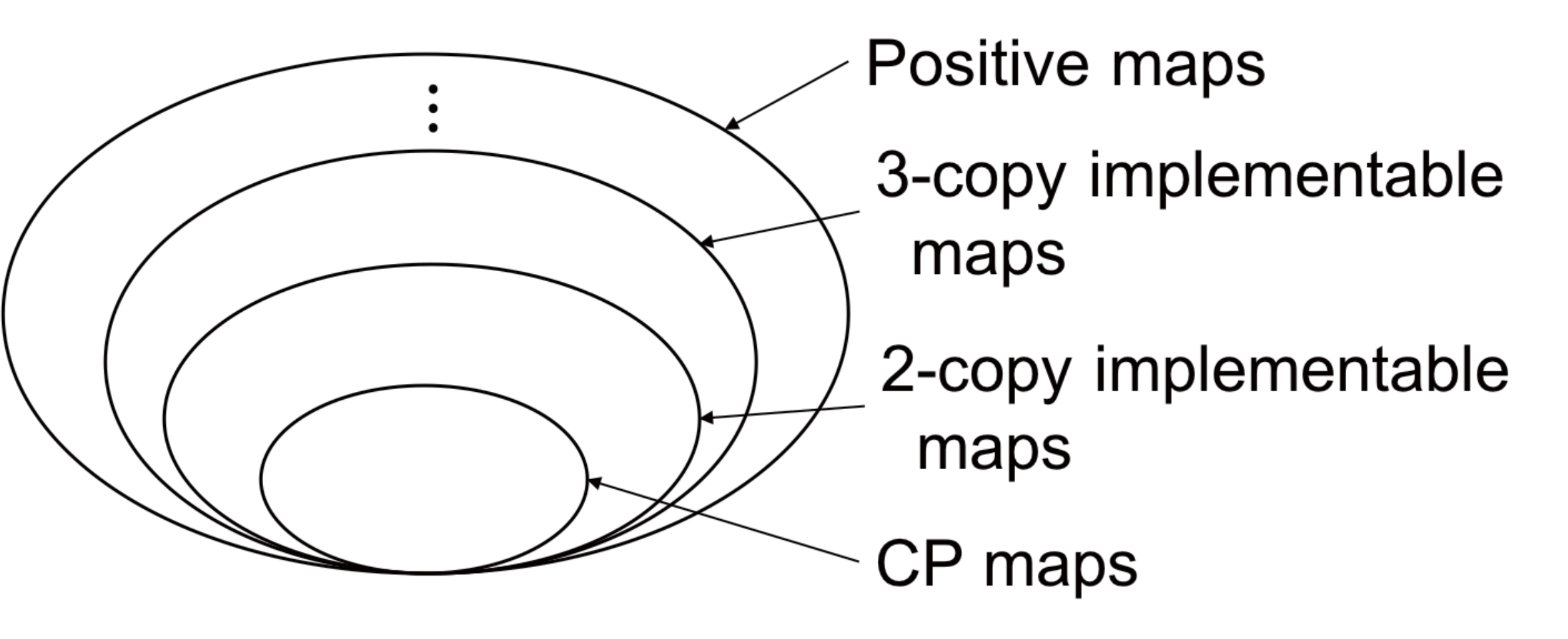}
\flushleft{(b)}
\includegraphics[width=85mm]{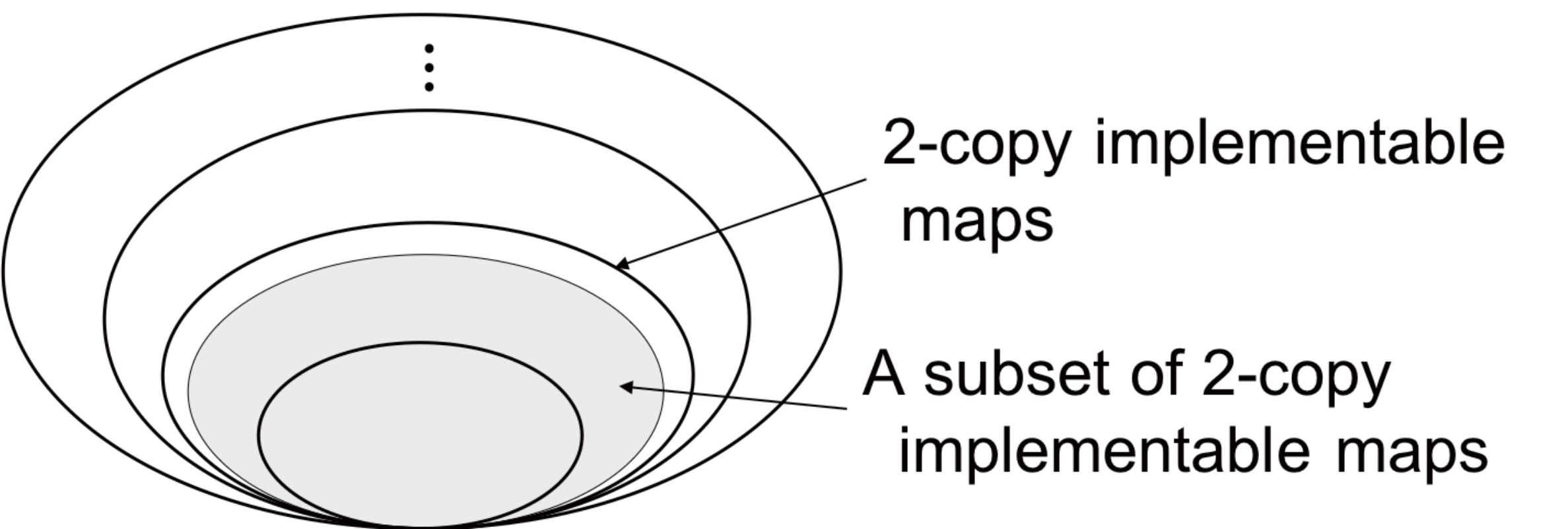}
\caption{(a) A pictorial representation of the set of positive maps.
Theorem~\ref{thm:optimality_of_sym_ext} gives a complete characterization of $N$-copy implementable maps.
While we have the inclusion relations between the sets,
it is not clear from Theorem~\ref{thm:optimality_of_sym_ext}
if the set of $N$-copy implementable maps converges to the set of positive maps, or how it converges.
(b) In Theorem~\ref{thm:n_copy_extension_add_white_noise}, 
we analyze the set of $N$-copy implementable maps by restricting the positive map to be
 the convex sum of another positive map with the depolarizing channel.
With this restriction, we obtain a subset of $N$-copy implementable maps.
We show that this subset of $N$-copy implementable maps converges to the set of positive maps with $N$.
This result also implies the set of $N$-copy implementable maps converges to the set of positive maps.
}\label{fig:set_representation}
\end{figure}

To show a range of $\eta$ that $\Lambda^\eta_N$ is CP for either case,
 we show the corresponding Choi operator $\JL^\eta_N$ is positive,
that is, the minimum eigenvalue of $ \JL^\eta_N $ is non-negative,
namely, $ \lambda_{\min} ( \JL^\eta_N) \geq 0$.
We first use Lemma~\ref{lem:purify} and~\ref{thm:navascues} to prove Lemma~\ref{lem:navascues_ext},
and then prove Theorem~\ref{thm:n_copy_extension_add_white_noise} and \ref{thm:implementable_eta}.

\begin{lemma}[Ref.~\cite{christandl07,math_entanglement}]\label{lem:purify}
If $\rho_{N} \in \Lin(\Hil_{01\cdots N}) $ is permutation invariant on the permutation of $\Hil_1, \ldots, \Hil_N$, 
i.e., $(\Id_0 \otimes P_\sigma) \rho_{N} (\Id_0 \otimes P_\sigma)^\dag = \rho_{N}$
for all permutation operators $P_\sigma$ of $\Hil_{1} \otimes \cdots \otimes \Hil_{N}$,
its purification $\ket{\psi_N} \in \Hil_{01\cdots N} \otimes \Hil_{0'1'\cdots N'}$, where $\Hil_{k'} \cong \Hil_k$, is Bose-symmetric on 
permutation of $\Hil_{11'}, \ldots, \Hil_{NN'}$,
 i.e., $(\Id_{00'} \otimes P_{\sigma'}) \ket{\psi_N} = \ket{\psi_N}$ for all permutation operators $P_{\sigma'}$ of $\Hil_{11'} \otimes \cdots \otimes \Hil_{NN'}$.
\end{lemma}

\begin{proof}
First diagonalize $\rho_N$ as $\rho_N = \sum_\lambda \lambda \Pi_\lambda$,
where $\Pi_\lambda$ is the projector onto the eigenspace corresponding to an eigenvalue $\lambda$.
For each $\lambda$,
\begin{align}
\Pi_\lambda = (\Id_0 \otimes P_\sigma) \Pi_\lambda (\Id_0 \otimes P_\sigma)^\dag
\end{align}
holds for all $P_\sigma$.
Defining $M := \sum_\lambda \sqrt{\lambda} \Pi_\lambda$,
$M$ also satisfies 
\begin{align}
M = (\Id_0 \otimes P_\sigma) M (\Id_0 \otimes P_\sigma)^\dag
\end{align}
for all $P_\sigma$.
Define a purification of $\rho_N$ as $\ket{\psi_N} := (M\otimes \Id) \ket{\phi^+}$,
where $\ket{\phi^+}$ is the maximally entangled state between $\Hil_{01\cdots N}$ and $\Hil_{0' 1' \cdots N'}$.
Then it satisfies 
\begin{align}
(\Id_{00'} \otimes P_{\sigma'}) \ket{\psi_N} &=(\Id_0 \otimes P_\sigma)\otimes(\Id_0 \otimes P_\sigma)\ket{\psi_N} \notag \\
&= [(\Id_0 \otimes P_\sigma)\otimes(\Id_0 \otimes P_\sigma)] (M\otimes \Id) \ket{\phi^+} \notag \\
&= (M\otimes \Id) [(\Id_0 \otimes P_\sigma)\otimes(\Id_0 \otimes P_\sigma) ]\ket{\phi^+} \notag \\
&= (M\otimes \Id) [(\Id_0 \otimes P_\sigma)(\Id_0 \otimes P_\sigma)^T \otimes \Id ]\ket{\phi^+} \notag \\
&=  (M\otimes \Id) \ket{\phi^+}  = \ket{\psi_N},
\end{align}
that is, $\ket{\psi_N}$ is Bose-symmetric on $\Hil_{11'} , \ldots, \Hil_{NN'}$.
\end{proof}

\begin{lemma}[Theorem~2 of Ref.~\cite{symmetric_extension}]\label{thm:navascues}
If $\rho \in \Lin(\Hil_{01})$ has an $N$ Bose-symmetric extension $\rho_N \in \Lin(\Hil_{01\cdots N})$ extending $\Hil_1$, that is,
\begin{enumerate}
\item
$\rho_{N}\, \geq 0$
\item
$\Tr_{2\cdots N}\, \rho_{N} = \rho$
\item
$\rho_{N}$ is Bose-symmetric, i.e., $(\Id_0 \otimes P_\sigma) \rho_{N} = \rho_{N}$
for all permutation operators $P_\sigma$ of $\Hil_{11'} \otimes \cdots \otimes \Hil_{NN'}$,
\end{enumerate}
then $\rho^\eta := (1-\eta) \rho + \eta (\Tr_1 \rho) \otimes \Id/d_1$ is separable for $\eta \geq d_1 / (N + d_1)$.
\end{lemma}

\begin{lemma}\label{lem:navascues_ext}
If $\rho \in \Lin(\Hil_{01})$ has an $N$ permutation invariant extension $\rho_N \in \Lin(\Hil_{01\cdots N}) $ extending $\Hil_1$,
that is,
\begin{enumerate}
\item
$\rho_{N}\, \geq 0$
\item
$\Tr_{2\cdots N}\, \rho_{N} = \rho$
\item
$\rho_{N}$ is permutation invariant, i.e., $(\Id_0 \otimes P_\sigma) \rho_{N} (\Id_0 \otimes P_\sigma)^\dag = \rho_{N}$
for all permutation operators $P_\sigma$ of $\Hil_{11'} \otimes \cdots \otimes \Hil_{NN'}$,
\end{enumerate}
then 
\begin{enumerate}
\item
$\rho^{\eta_b} = (1-\eta_b) \rho + \eta_b \Tr_1 \rho \otimes \Id/d_1$ is separable for $\eta_b \geq d_1^2 / (N + d_1^2)$.
\item
$\rho^{\eta_a} = (1-\eta_a) \rho + \eta_a \Id/d_0 \otimes \Id/d_1$ is separable for $\eta_a \geq d_0 d_1^2 / (N + d_0 d_1^2)$.
\end{enumerate}
\end{lemma}

\begin{proof}
1. 
Lemma~\ref{lem:purify} shows that $\rho_N$'s purification $\ket{\psi_N} \in \Hil_{01\cdots N} \otimes \Hil_{0'1'\cdots N'}$ is Bose-symmetric.
Let $\rho' = \Tr_{22' \cdots NN'} \proj{\psi_N}$,
then $\rho'$ has an $N$ Bose-symmetric extension and $\rho = \Tr_{0'1'} \rho'$.
Lemma~\ref{thm:navascues} shows that 
$(1-\eta_b) \rho' + \eta_b \Tr_{11'} \rho' \otimes \Id/d_1^2$ is separable between 
$\Hil_{00'}$ and $\Hil_{11'}$ for $\eta_b \geq d_1^2 / (N + d_1^2)$, where $d_1^2 = \dim \Hil_{11'}$.
By taking partial trace on $\Hil_{0'1'}$, we show that $(1-\eta_b) \rho + \eta_b \Tr_1 \rho \otimes \Id/d_1$ is also separable.

2. 
Let $\rho_0 = \Tr_1 \rho$ be decomposed as $\sum_{k=0}^{d_0-1} p_k \proj{\psi_k}$, 
and $(\rho^0)^c_i = \sum_{k=0}^{d_0-1} p_{(k+i)\ \mathrm{mod}\ d} \proj{\psi_k}$ for $i=1,\ldots,d_0-1$,
then $\rho_0 + \sum_{i=1}^{d_0 - 1} (\rho^0)^c_i = \Id$.
Thus, $\rho^{\eta_a}$ can be written with $\rho^{\eta_b}$ by
\begin{align}
\rho^{\eta_a} = p \rho^{\eta_b} + \frac{1-p}{d_0-1} \sum_{i=1}^{d_0 - 1} (\rho^0)^c_i \otimes \frac{\Id_1}{d_1}
\end{align}
with $1-\eta_a = p(1-\eta_b)$ and $p$ satisfies $p \eta_b = (1-p) / (d_0 - 1)$.
Since $\rho^{\eta_b}$ is separable for $\eta_b \geq d_1^2 / (N + d_1^2)$, 
 $\rho^{\eta_a}$ is also separable for $\eta_a \geq d_0 d_1^2 / (N + d_0 d_1^2)$
\end{proof}

\begin{proof}[Proof of Theorem~\ref{thm:n_copy_extension_add_white_noise} and \ref{thm:implementable_eta}]
We first  prove Theorem~\ref{thm:implementable_eta}.
It is enough to prove that 
the minimum eigenvalue of $ \JL^{\eta_b}_N $ is non-negative for $\eta_b \geq d_1^2 / (N + d_1^2)$.
$ \JL^{\eta_b}_N $ is invariant under permutations of $\mathcal{H}_1,\ldots,\mathcal{H}_N$,
that is, for any permutation operator $P_\sigma$,
\begin{align}
(\Id_0 \otimes P_\sigma) \JL^{\eta_b}_N (\Id_0 \otimes P_\sigma)^\dag = \JL^{\eta_b}_N
\end{align} holds.
By considering this symmetry, $ \JL^{\eta_b}_N $ can be written as
\begin{align}
\JL^{\eta_b}_N 
= \frac{1}{N!} \sum_{\sigma} (\Id_0 \otimes P_\sigma)  ( \JL^{\eta_b} )_{01} \otimes \Id_{R} (\Id_0 \otimes P_\sigma)^\dag,
\end{align}
where $\JL^{\eta_b}$ is the Choi operator of $\Lambda^{\eta_b}$.

Consider a unit vector $\ket{\phi} \in \Hil_{01\cdots N}$.
The smallest eigenvalue of $ \JL^{\eta_b}_N $ is given by 
\begin{align}
\lambda_{\min}( \JL^{\eta_b}_N ) = \min_{ |\!| \ket{\phi} |\!| = 1} \bra{\phi} \JL^{\eta_b}_N \ket{\phi},
\end{align}
where
\begin{align}
&\bra{\phi} \JL^{\eta_b}_N \ket{\phi} \notag \\
&= \Tr \JL^{\eta_b}_N \proj{\phi} \notag \\ 
&= \Tr  \frac{1}{N!} \sum_{\sigma} (\Id_0 \otimes P_\sigma) ( (\JL^{\eta_b})_{01} \otimes \Id_{R}) (\Id_0 \otimes P_\sigma)^\dag \proj{\phi} \notag \\
&= \Tr ((\JL^{\eta_b})_{01} \otimes \Id_{R}) \sum_{\sigma}\frac{1}{N!} (\Id_0 \otimes P_\sigma) \proj{\phi} (\Id_0 \otimes P_\sigma)^\dag \notag \\
&=: \Tr ((\JL^{\eta_b})_{01} \otimes \Id_{R}) \rho_N \\
&=: \Tr (\JL^{\eta_b})_{01} \rho.
\end{align}
The permutation symmetry of $\JL^{\eta_b}_N$ is transformed to the permutation symmetry of the state $\rho_N$,
as $\rho_N$ satisfies 
\begin{align}
\rho_N = (\Id_0 \otimes P_\sigma) \rho_N (\Id_0 \otimes P_\sigma)^\dag
\end{align}
 for all $P_\sigma$.
Thus $\rho$ has an $N$ permutation invariant extension.
Here 
\begin{align}
\Tr \JL^{\eta_b} \rho &= (1-{\eta_b}) \Tr \JL \rho + {\eta_b} \Tr (\Tr_1 \JL \otimes \Id_1/d_1) \rho \notag \\
&= (1-{\eta_b}) \Tr \JL \rho + {\eta_b} \Tr (\Tr_1 \JL ) (\Tr_1 \rho) / d_1 \notag \\
&= (1-{\eta_b}) \Tr \JL \rho + {\eta_b} \Tr \JL  (\Tr_1 \rho \otimes \Id_1 / d_1) \notag \\
&= \Tr \JL \rho^{\eta_b}
\end{align}
holds, and the first bound of Lemma~\ref{lem:navascues_ext} shows that $\rho^{\eta_b}$ is separable for ${\eta_b} \geq d_1^2 / (N + d_1^2)$.
Notice that for a positive map $\Lambda$,
the corresponding Choi operator $\JL$ satisfies
\begin{align}
\Tr \JL (\rho \otimes \sigma) \geq 0
\end{align}
for all positive operators $\rho $ and $\sigma$~\cite{positive_maps}.
Therefore, $\lambda_{\min}(\JL^{\eta_b}_N) \geq 0$ for these ${\eta_b}$.

The proof of Theorem~\ref{thm:n_copy_extension_add_white_noise} is in a similar way, 
and the only different part is that
\begin{align}
\Tr \JL^{\eta_a} \rho &= (1-{\eta_a}) \Tr \JL \rho + {\eta_a} \frac{\Tr \JL}{d_1} \frac{\Tr (\Id_0 \otimes \Id_1) \rho }{d_0} \notag \\
&= (1-{\eta_a}) \Tr \JL \rho + {\eta_a} (\Tr\JL) / d_0 d_1 \notag \\
&= (1-{\eta_a}) \Tr \JL \rho + {\eta_a} \Tr \JL (\Id_0/d_0 \otimes \Id_1 / d_1) \notag \\
&= \Tr \JL \rho^{\eta_a}.
\end{align}
The second bound of Lemma~\ref{lem:navascues_ext} shows Theorem~\ref{thm:n_copy_extension_add_white_noise}.

\end{proof}

Note that for $d_1 = 2$, the two bounds shown in Lemma~\ref{lem:navascues_ext} can be improved using the result of~\cite{symmetric_two_dim}.
Under the condition of Lemma~\ref{lem:navascues_ext}, it is shown in~\cite{symmetric_two_dim} that 
$\rho^{\eta_b} = (1-\eta_b) \rho + \eta_b \Tr_1 \rho \otimes \Id/d_1$ 
is separable for $\eta_b \geq d_1 / (N + d_1)$,
which corresponds to the first bound.
Following the same calculation in the proof of Lemma~\ref{lem:navascues_ext}, we also obtain that 
$\rho^{\eta_a} = (1-\eta_a) \rho + \eta_a \Id/d_0 \otimes \Id/d_1$ is separable for $\eta_a \geq d_0 d_1 / (N + d_0 d_1)$.
These results also apply to Theorem~\ref{thm:n_copy_extension_add_white_noise} and \ref{thm:implementable_eta},
and the bounds can be improved in the same way for $d_1 = 2$.

\section{Positive maps with no $N$-copy extension}
\label{sec:no_n_copy_extension}

Theorem~\ref{thm:n_copy_extension_add_white_noise} gives a sufficient condition for a map to be $N$-copy implementable,
but it cannot be used to show that a map is not $N$-copy implementable.
In this section, we give a necessary condition for a map to be $N$-copy implementable.
We show a theorem to bound the set of $N$-copy implementable maps complementary to Theorem~\ref{thm:n_copy_extension_add_white_noise} and Theorem~\ref{thm:implementable_eta}.
With this theorem, we also provide  a characterization of certain  types of maps that are not finite-copy implementable in the next section.

\begin{theorem}\label{thm:condition_for_no_n_copy_extension}
$\Lambda$ is not $N$-copy implementable if the following operator is not positive.
\begin{align}
\sum_{i,j=0}^{d_1-1}   \ketbra{i}{j} \otimes \Lambda( \ketbra{i}{j} ) 
+ (N-1) \sum_{i=1}^{d_1-1}  \proj{i}  \otimes \Lambda(\proj{0})  \label{eq:condition_for_no_n_copy_extension}
\end{align}
\end{theorem}

The first term of Eq.~\eqref{eq:condition_for_no_n_copy_extension} is the Choi operator of $\Lambda$,
and the second term contains $\Lambda(\proj{0})$. 
We take a state $\ket{0}$ in the computational basis $\{ \ket{i} \}_0^{d_1-1}$ for convenience, 
 but this condition can be applied for any orthonormal basis $\{ \ket{k_i} \}$, namely,
if $\JL + (N-1)\sum_{i \neq 0} \Lambda(\proj{k_0}) \otimes \proj{k_i}$ is not positive,
$\Lambda$ is not $N$-copy implementable.

 If we require $\Lambda_N(\rho^{\otimes N}) = \Lambda(\rho)$ for any density operator $\rho$,
the optimal construction of $N$-copy extension is already given by Eq.~\eqref{eq:n_copy_extension}.
In this case, Theorem~\ref{thm:condition_for_no_n_copy_extension} can be proved by showing that the operator Eq.~\eqref{eq:condition_for_no_n_copy_extension} is the result of a certain positive map applied on the Choi operator of the $N$-copy extension Eq.~\eqref{eq:n_copy_extension}.
However, we will give a proof of Theorem~\ref{thm:condition_for_no_n_copy_extension}, 
which only uses constraints for $\mathrm{span}\{ \proj{\psi}^{\otimes N} \}$,
where
\begin{align}
\mathrm{span}\{ \proj{\psi}^{\otimes N} \} := 
\left\{ \sum_i c_i \proj{\psi_i}^{\otimes N} \middle| c_i \in \mathbb{C}, \ket{\psi_i} \in \mathbb{C}^{d_1} \right\}.
\end{align}
Since for $N \geq 2$, $\mathrm{span}\{ \proj{\psi}^{\otimes N} \} \subsetneq \mathrm{span}\{ \rho^{\otimes N} \} $ holds\footnote{For example, for $d=2,N=2$, $\Id^{\otimes 2}$ is in $\mathrm{span}\{ \rho^{\otimes 2} \}$, 
but not in $ \mathrm{span}\{ \proj{\psi}^{\otimes 2} \} $.
That is, for any state $ O \in \mathrm{span}\{ \proj{\psi}^{\otimes 2} \} $, it can be decomposed as $O = \sum_i c_i \proj{\psi_i}^{\otimes 2}$, and since $\bra{\psi^-} (\proj{ \psi_i }^{\otimes 2}) \ket{\psi^-} = 0$ for any $\ket{\psi_i}$, 
we obtain $ \bra{\psi^-} O \ket{\psi^-} = 0 \neq \bra{\psi^-} (\Id^{\otimes 2}) \ket{\psi^-} $ where $\ket{\psi^-}$ is the singlet state.
},
our proof indicates that Theorem~\ref{thm:condition_for_no_n_copy_extension} can be applied for a more general problem.
That is, even if we require $\Lambda_N(\rho^{\otimes N}) = \Lambda(\rho)$ for only pure state $\rho=\proj{\psi}$,
it is not $N$-copy implementable if the operator given by Eq.~(\ref{eq:condition_for_no_n_copy_extension}) is not positive.

\begin{proof}[Sketch of the proof]

Let $\JL_N$ be the Choi operator of $\Lambda_N$, an arbitrary $N$-copy extension of $\Lambda$.
Since $\Lambda_N$ is only defined on a subspace of $\Lin(\Hil_{1 \cdots N})$,
not all elements of $\JL_N$ are determined.
However, $\JL_N$ cannot be positive if there exists a positive map which maps $\JL_N$ to a non-positive operator.
We explicitly construct the positive map that maps $\JL_N$ to the operator given by Eq.~\eqref{eq:condition_for_no_n_copy_extension} in Appendix~\ref{ap:proof_of_no_n_copy_extension}.

\end{proof}

\section{Boundary and extremal conditions}\label{sec:boundary_extremal}

All positive maps that are non-boundary are finite-copy implementable 
as  shown in Corollary~\ref{cor:non_boundary_implementable}.
A natural question is that  whether the boundary condition can fully characterize finite-copy implementability.
We show that this is not the case.
Among boundary positive maps, there exist both finite-copy implementable ones and not implementable ones.
Meanwhile, boundary positive maps can be further classified in to extremal positive maps and non-extremal ones.
We conjecture that all extremal PNCP maps are not finite-copy implementable,
whereas there exist both finite-copy implementable and not implementable non-extremal positive maps.

We first see the examples for extremal positive maps.
For the 2-dimensional case, any extremal PNCP map is known to be the concatenation of the transposition map and a CP map with a single Kraus operator~\cite{positive_maps}.
Notice that if the Kraus operator is rank 0 or 1, the concatenation with the transposition map is CP,
so the Kraus operator is necessarily full rank, and the CP map is invertible.
Thus, since the transposition map is not finite-copy implementable~\cite{universal_not,certifying,miyazaki},
which also follows from Theorem~\ref{thm:condition_for_no_n_copy_extension},
any 2-dimensional extremal PNCP map is not finite-copy implementable.
For higher dimensional maps, there exist extremal positive  maps that cannot be decomposed into the transposition map and a CP map.
One such example is the Choi map~\cite{choimap,positive_maps}, which is a map between 3-dimensional matrices as
\begin{align}
\mathcal{C} : 
\begin{pmatrix}
x_{00} & x_{01} & x_{02} \\
x_{10} & x_{11} & x_{12} \\
x_{20} & x_{21} & x_{22}
\end{pmatrix}
\mapsto 
\begin{pmatrix}
x_{00} + x_{22} & -x_{01} & -x_{02} \\
-x_{10} & x_{00} + x_{11} & -x_{12} \\
-x_{20} & -x_{21} & x_{11} + x_{22}
\end{pmatrix}.
\end{align}

The Choi map is also shown to be not finite-copy implementable by applying Theorem~\ref{thm:condition_for_no_n_copy_extension}.
It is enough to show that the corresponding operator $\JL'_N$ is not positive.
Consider the principle minor of $\JL'_N$ given by
\begin{align}
&\begin{pmatrix}
\bra{00}\JL'_N\ket{00} & \bra{00}\JL'_N\ket{11} & \bra{00}\JL'_N\ket{22} \\
\bra{11}\JL'_N\ket{00} & \bra{11}\JL'_N\ket{11} & \bra{11}\JL'_N\ket{22} \\
\bra{22}\JL'_N\ket{00} & \bra{22}\JL'_N\ket{11} & \bra{22}\JL'_N\ket{22}
\end{pmatrix}
\\&\qquad=
\begin{pmatrix}
1 & -1 & -1 \\
-1 & 1 + (N-1) & -1 \\
-1 & -1 & 1
\end{pmatrix}.
\end{align}
The determinant of this matrix is $-4$ independent of $N$.
Thus, for any $N$, this principle minor of $\JL'_N$ is not positive,
and $\JL'_N$ is not positive.
Note that even if the determinant does not depend on $N$, the eigenvalues depend on $N$.
Especially, the smallest eigenvalue converges to 0 for $N \to \infty$, and the noisy version is always finite-copy implementable.

Next, we consider two non-extremal examples, one is finite-copy implementable and the other is not.
The implementable one is given by
$
\Lambda(\rho) := (1-p)\rho + \frac{1}{2} p\, \mathcal{C}(\rho),
$
where $p \in ( 6/7 , 8/9 )$.
The smallest eigenvalue of the Choi operator of $\Lambda$ is $-(7p-6)/2$,
and thus for $p > 6/7$, $\Lambda$ is not CP.
However, to construct the $2$-copy extension given by Eq.~\eqref{eq:n_copy_extension},
the minimum eigenvalue of the corresponding Choi operator is non-negative for $p<8/9$,
and thus $\Lambda$ is 2-copy implementable.
 The example of the finite-copy  non-implementable map is given by
$
\Lambda(\rho) := (1-p) \rho + p \rho^T
$
with $p \in (0,1)$.
We apply Theorem~\ref{thm:condition_for_no_n_copy_extension} for this map.
It is enough to show that the corresponding operator $\JL'_N$ is not positive.
Consider the principle minor given by
\begin{align}
\begin{pmatrix}
\bra{01} \JL'_N \ket{01} & \bra{01} \JL'_N \ket{10} \\
\bra{10} \JL'_N \ket{01} & \bra{10} \JL'_N \ket{10}
\end{pmatrix}
=
\begin{pmatrix}
0 & p \\
p & N-1
\end{pmatrix}.
\end{align}
This matrix is not positive for $p \neq 0$.
Thus, $\Lambda$ does not have a CP $N$-copy extension for any $N$,
and is not finite-copy implementable.

\section{Optimal state transposition}\label{sec:transposition}

As the transposition map is one of the most important  PNCP maps,
we consider an optimal approximation to implement  the map with $N$ copies of an input  state.
That is, we consider the implementation of a noisy version of  the transposition map
\begin{align}
T^\eta : \rho \mapsto (1-\eta) \rho^T + \eta (\Id / d) \Tr(\rho), \label{eq:noisy_transposition}
\end{align}
and show the amount of white noise $\eta$ required to implement the map when $N$ copies of an input are provided.
In this section, we will show a lower bound for $\eta$.
An upper bound for $\eta$, on the other hand, is given by Theorem~\ref{thm:implementable_eta}.
As the two bounds coincide in the asymptotic limit, 
we can see that the bound shown in Theorem~\ref{thm:implementable_eta} is also asymptotically optimal.

\begin{theorem}\label{thm:implementable_transpose}
For $d$-dimensional noisy transposition $T^\eta : \rho \mapsto (1-\eta) \rho^T + \eta (\Id / d ) \Tr(\rho)$,
it is $N$-copy implementable for $\eta \geq {d^2}/({N+d^2})$,
and is not $N$-copy implementable for $\eta < \min \{ d/(d+1)  ,d (d-1) / (N + d(d-1) ) \}$.
\end{theorem}

This result shows that for $d$-dimensional state transposition,
if we assume the approximation is given as Eq.~\eqref{eq:noisy_transposition},
multiple copies of an input does not help unless the number of copies is more than $d-1$.
Numerical calculations shows the bound $\min \{ d/(d+1) , d(d-1) / (N + d(d-1) ) \}$ seems to be 
the critical amount of the noise required to make the noisy transposition  to be $N$-copy implementable.

\begin{proof}[Sketch of the proof]

Since the Choi operator of state transposition is the swap operator $S$,
the Choi operator of the optimal $N$-copy extension is given as
\begin{align}
\JL_N = \frac{1}{N} \sum_{i=1}^N S_{0i} \otimes \Id_{R}.
\end{align}
We show that the minimum eigenvalue of $\JL_N$ is bounded as
\begin{align}
\lambda_{\min} (\JL_N) \leq
\begin{cases}
-1 & N \leq d-2 \\
-(d-1)/N & N \geq d-1
\end{cases}. \label{eq:min_eigen_optimal_transpose}
\end{align}
The $N=1$ case is easy to prove as $\JL_1 = S_{01}$, which is just the swap operator.
If $T^\eta$ is $N$-copy implementable, it is also $N'$-copy implementable for $N' > N$,
thus the minimum eigenvalue of $\JL_N$ is non-decreasing in $N$.
Therefore, for $N \leq d-2$ cases, it is enough to show that the minimum eigenvalue of $\JL_{d-1}$ is $-1$.
We will only show for $N \geq d-1$ cases.
We give an explicit construction of (unnormalized) eigenvector $\ket{\Psi_N}$ 
which has the eigenvalues shown in Eq.~\eqref{eq:min_eigen_optimal_transpose} in Appendix~\ref{ap:proof_of_implementable_transpose}.

\end{proof}

\section{Characterization of positive maps based on $N$-copy implementability}\label{sec:characterization}

Positive maps provide a useful mathematical tool to understand and detect entanglement. 
In particular, positive maps have a one-to-one relation with entanglement witness 
and every PNCP map  can be used to certify the entanglement of some quantum state~\cite{ppt2}. 
The problem of deciding whether a linear map is positive is known to be computationally hard~\cite{gurvits04} 
and methods to characterize positive maps that does not follow directly from the definition date back to pioneering studies by Jamio\l kowski~\cite{jamiolkowski72}.

Although it is not the main focus of this paper, 
our results imply a characterization of strictly positive maps\footnote{
A linear map $\Lambda:\mathcal{L}(\mathcal{H}_1)\to \mathcal{L}(\mathcal{H}_0) $ is strictly positive 
if it transforms positive operators into strictly positive operators. 
More precisely if $\rho \in \mathcal{L}(\mathcal{H}_1)$ respects $\bra{\psi_1}\rho\ket{\psi_1} \geq 0$ for every $\ket{\psi_1}\in \mathcal{H}_1$, 
strictly positive maps $\Lambda$ are the ones which respect $\bra{\psi_0} \Lambda(\rho) \ket{\psi_0}>0$ for every $\ket{\psi_0}\in \mathcal{H}_0$. 
It also follows straightforwardly a linear map is strictly positive map if and only if it is in the interior of the set of positive maps.
}.
Clearly, if a map is not  positive, it is not $N$-copy implementable 
since that would necessarily transform a positive operator into a non-positive one. 
Also, Corollary~\ref{cor:non_boundary_implementable} states that all strictly positive maps are $N$-copy implementable for finite $N$. 
Since Theorem~\ref{thm:optimality_of_sym_ext} can be used to decide whether 
a positive map has a CP $N$-copy extension by checking if the minimum eigenvalue of a linear operator is positive, 
we have a hierarchy that can always decide if a map is strictly positive. 
This hierarchy also quantifies the ``non-complete-positiveness'' of positive maps
by the number of copies, $N$, that is required to implement it via a CP-map. 
Also, Theorems~\ref{thm:n_copy_extension_add_white_noise} and~\ref{thm:implementable_eta} provide universal bounds 
on the number of copies $N$ that only depend on the dimension of the input and output space.

We note that the operators which have a symmetrized structure as in Eq.~\eqref{eq:n_copy_extension_choi} have been investigated~\cite{universal_not,cloning_werner,separability,math_entanglement,lancien2016k},
and a characterization of strict positive maps based on minimum eigenvalues of such operators has also appeared previously in  Section X of~\cite{separability}. 
In~\cite{separability}, the authors exploit the relation between positive maps and entanglement witness to present an SDP hierarchy based on the symmetric extension of quantum states~\cite{separability_prl}. 
Theorem 2 establishes that the operator studied in Section X of Ref.~\cite{separability} happens to characterize $N$-copy implementability of a given positive map. 
Thus, Theorem 3 and 4 can also be used to quantitatively show the convergence of the hierarchy in Ref.~\cite{separability}.
On the other hand, combining Theorem~3 of Ref.~\cite{separability} with our Theorem 2, we can obtain Corollary~\ref{cor:non_boundary_implementable}.
We emphasize that the concept of $N$-copy implementability does not appear in previous researches.

\section{Conclusion}

We have introduced the concept of $N$-copy extension to analyze the implementability of positive  but not completely positive (PNCP) maps by consuming multiple copies of an input state.  
We gave a universal construction of $N$-copy extension on the whole space.
This construction is shown to be  optimal in the sense that
a map is $N$-copy implementable if and only if it is implementable with this construction.
This result implies that the $N$-copy implementability is verifiable 
by calculating the smallest eigenvalue of the Choi operator of this $N$-copy extension.
For a simpler verification of $N$-copy implementability,
we also gave two sufficient conditions and one necessary condition.
Moreover, we showed all non-boundary positive maps are finite-copy implementable.
We finally analyzed the $N$-copy implementability of  a noisy version of the transposition map,
and  gave an asymptotically optimal value of the noise required to add  for the noisy version of the map to be $N$-copy implementable. 
The concept of $N$-copy implementability also gives a quantitative characterization for positive maps 
which provides an operational meaning to the characterization of strictly positive maps presented in Section X of~\cite{separability}.

We have shown that non-boundary positive maps are always finite-copy implementable,
but the behavior of boundary positive maps  is still not fully understood.
We  conjecture that any extremal PNCP maps requires infinite copies of an input to be implementable.
The existence of any simpler necessary or sufficient condition for finite-copy implementability remains open.

\section*{Acknowledgements}

We thank S.~Nakayama for helpful discussions.
This work was supported by ALPS and JSPS KAKENHI (Grant No. 15H01677, 16H01050, 17H01694, 18H04286, 16F16769 and 18K13467), and the Q-LEAP project of MEXT, Japan.

\bibliography{refs}

\appendix

\section{Proof of Theorem~\ref{thm:condition_for_no_n_copy_extension}}\label{ap:proof_of_no_n_copy_extension}

Let $\JL_N$ be the Choi operator of $\Lambda_N$, an arbitrary $N$-copy extension of $\Lambda$.
Since $\Lambda_N$ is only defined on a subspace of $\Lin(\Hil_{1 \cdots N})$,
not all elements of $\JL_N$ are determined.
However, $\JL_N$ cannot be positive if there exists a positive map which maps $\JL_N$ to a non-positive operator.

Let $\Phi : \Lin( \Hil_{0 1  \cdots  N} ) \to \Lin( \Hil_{01})$ be a positive map defined as $\Phi(\rho) = V \rho V^\dag$ 
where $V \in \Lin( \Hil_{0 1  \cdots  N} , \Hil_{0 1} )$ is
\begin{align}
V &= \ket{0}_1 (\bra{0}_1 \otimes \bra{0}_2 \otimes \cdots \otimes \bra{0}_{N}) \otimes \Id_0 \notag\\
&\quad + \sum_{i=1}^{d_1 - 1} \sum_{k=1}^{N} \ket{i}_1 (\bra{i}_1 \otimes \bra{0}_2 \otimes \cdots \otimes \bra{0}_{N}) S_{1k}
\otimes \Id_0,
\end{align}
 $S_{kl}$ is the swap operator defined as
$S_{kl} \ket{\psi}_k \ket{\phi}_l = \ket{\phi}_k \ket{\psi}_l$ for all $\ket{\psi},\ket{\phi} \in \mathbb{C}^{d_1}$,
and  subscripts $k$ and $l$ denote Hilbert space $\Hil_k$ and $\Hil_l$, respectively.

This map $\Phi$ maps $\JL_N$ to the operator, 
\begin{align}
\JL'_N :=& \Phi(\JL_N ) \\
=&  \ketbra{0}{0} \otimes \Lambda_N( a_{00} ) 
+ \sum_{j=1}^{d_1-1}  \ketbra{0}{j} \otimes \Lambda_N( a_{0j} )  \notag\\
& +\sum_{i=1}^{d_1-1}  \ketbra{i}{0} \otimes \Lambda_N( a_{i0} ) 
+ \sum_{i,j=1}^{d_1-1}  \ketbra{i}{j} \otimes \Lambda_N( a_{ij} ) 
\end{align}
with
\begin{align}
a_{00} &= \proj{0}^{\otimes N}, \label{aij_1}\\
a_{0j} &= \sum_{k=1}^{N} S_{1k} (\ketbra{0}{j}_1 \otimes \proj{0}^{\otimes (N-1)}_R) S^\dag_{1k}, \\
a_{i0} &= \sum_{k=1}^{N} S_{1k} (\ketbra{i}{0}_1 \otimes \proj{0}^{\otimes (N-1)}_R) S^\dag_{1k}, \\
a_{ij} &= \sum_{k=1}^{N} S_{1k} (\ket{i}_1 \otimes \ket{0}^{\otimes (N-1)}_R) 
\sum_{l=1}^{N} (\bra{j}_1 \otimes \bra{0}^{\otimes (N-1)}_R)S^\dag_{1l}, \label{aij_4}
\end{align}
and $\{ a_{ij} \}$ are permutation invariant on the Hilbert space $\Hil_1 \otimes \ldots \otimes \Hil_N$.

We prove that the operators $\{ a_{ij} \} $ defined by Eq.~\eqref{aij_1} - Eq.~\eqref{aij_4} 
satisfy $a_{ij} \in \mathrm{span}\{ \proj{\psi}^{\otimes N} \}$ by construction.
As $a_{00}$ is a density operator,
$a_{00} \in \mathrm{span}\{ \proj{\psi}^{\otimes N} \}$ is trivial.
Since only the linear span is considered, normalization factors are ignored for simplicity.

We first show that for any $\ket{\phi_1}$ and $\ket{\phi_2}$, 
$\ketbra{\phi_1}{\phi_2}^{\otimes N} \in \mathrm{span}\{ \proj{\psi}^{\otimes N} \}$,
which is also shown in \cite{church_of_sym}.
Consider a vector defined by $\ket{\psi_\theta} = \ket{\phi_1} + e^{i\theta} \ket{\phi_2}$, 
$\proj{\psi_\theta}$ is a density operator and
$\proj{\psi_\theta}^{\otimes N} \in \mathrm{span}\{ \proj{\psi}^{\otimes N} \}$. 
Since $\int_0^{2\pi} d\theta \, e^{ik\theta} = 0$ unless $k = 0$,
we obtain
\begin{align}
\ketbra{\phi_1}{\phi_2}^{\otimes N} =& \frac{1}{2\pi} \int_0^{2\pi} d\theta\, e^{iN\theta} \proj{\psi_\theta}^{\otimes N} \\
&\in \mathrm{span}\{ \proj{\psi}^{\otimes N} \}.
\end{align}

To prove $a_{i0} \in \mathrm{span}\{ \proj{\psi}^{\otimes N} \}$ for $i \neq 0$,
let $\ket{\psi_\theta} = \ket{0} + e^{i\theta} \ket{i}$,
then $\ketbra{\psi_\theta}{0}^{\otimes N} \in \mathrm{span}\{ \proj{\psi}^{\otimes N} \}$,
and thus
\begin{align}
a_{i0} = \frac{1}{2\pi} \int_0^{2\pi} d\theta\, e^{-i\theta} \ketbra{\psi_\theta}{0}^{\otimes N} \in \mathrm{span}\{ \proj{\psi}^{\otimes N} \}.
\end{align}
$a_{0j} \in \mathrm{span}\{ \proj{\psi}^{\otimes N} \}$ is also proved in the similar way.

To prove $a_{ij} \in \mathrm{span}\{ \proj{\psi}^{\otimes N} \}$ for $i,j \neq 0$, 
let $\ket{\psi_{1,\theta} } = \ket{0} + e^{i\theta} \ket{i}$ and $\ket{\psi_{2,\phi} } = \ket{0} + e^{-i\phi} \ket{j}$,
we have $\ketbra{\psi_{1,\theta}}{\psi_{2,\phi}}^{\otimes N} \in \mathrm{span}\{ \proj{\psi}^{\otimes N} \}$.
Since 
\begin{align}
\ketbra{\psi_{1,\theta}}{\psi_{2,\phi}} = \proj{0} + e^{i\theta}\ketbra{i}{0} + e^{i\phi}\ketbra{0}{j} + e^{i(\theta+\phi)}\ketbra{i}{j},
\end{align}
we obtain 
\begin{align}
a_{ij} &= \frac{1}{(2\pi)^2} \int_0^{2\pi} d\theta \int_0^{2\pi} d\phi\, e^{-i(\theta+\phi)} \ketbra{\psi_{1,\theta}}{\psi_{2,\phi}}^{\otimes N} \\
&\qquad \in \mathrm{span}\{ \proj{\psi}^{\otimes N} \}.
\end{align}

Since $\{ a_{ij} \} \in \mathrm{span}\{ \proj{\psi}^{\otimes N} \}$, 
 $\Lambda_N(a_{ij})$ are calculable with partial traces of $\{ a_{ij} \} $ as Eq.~\eqref{eq:calc_by_partial_trace}, that is,
\begin{align}
\Lambda_N(a_{00}) &= \Lambda(\proj{0}) , \\
\Lambda_N(a_{0j}) &= \Lambda( \ketbra{0}{j} ), \\
\Lambda_N(a_{i0}) &= \Lambda( \ketbra{i}{0} ), \\
\Lambda_N(a_{ij}) &= \Lambda( \ketbra{i}{j} ) + \delta_{ij} (N-1) \Lambda(\proj{0})
\end{align}
for $i,j \neq 0$.
Thus, we obtain
\begin{align}
\JL'_N &=    \ketbra{0}{0} \otimes \Lambda( \proj{0} )
+ \sum_{j=1}^{d_1-1} \ketbra{0}{j} \otimes \Lambda( \ketbra{0}{j} )  \notag\\
&\quad +\sum_{i=1}^{d_1-1}  \ketbra{i}{0} \otimes \Lambda( \ketbra{i}{0} ) 
+ \sum_{i,j=1}^{d_1-1}  \ketbra{i}{j}  \otimes \Lambda( \ketbra{i}{j} ) \notag\\
&\quad + (N-1) \sum_{i,j=1}^{d_1-1}  \delta_{ij} \ketbra{i}{j} \otimes \Lambda(\proj{0}) \\
&= \JL + (N-1) \sum_{i=1}^{d_1-1} \proj{i} \otimes \Lambda(\proj{0}) .
\end{align}
Therefore, 
if $\JL'_N = \JL + (N-1) \sum_{i=1}^{d_1-1} \proj{i} \otimes \Lambda(\proj{0}) $ is not positive,
$\JL_N$ cannot be positive,
that is, no CP $N$-copy extension of $\Lambda$ exists.

\section{Proof of Theorem~\ref{thm:implementable_transpose}}\label{ap:proof_of_implementable_transpose}

Since the Choi operator of state transposition is the swap operator $S$,
the Choi operator of the optimal $N$-copy extension is given as
\begin{align}
\JL_N = \frac{1}{N} \sum_{i=1}^N S_{0i} \otimes \Id_{R}.
\end{align}
We show that the minimum eigenvalue of $\JL_N$ is bounded as
\begin{align}
\lambda_{\min} (\JL_N) \leq
\begin{cases}
-1 & N \leq d-2 \\
-(d-1)/N & N \geq d-1
\end{cases}. \label{eqap:min_eigen_optimal_transpose}
\end{align}
The $N=1$ case is easy to prove as $\JL_1 = S_{01}$, which is just the swap operator.
If $T^\eta$ is $N$-copy implementable, it is also $N'$-copy implementable for $N' > N$.
Thus, the minimum $\eta$ that allows $T^\eta$ to be $N$-copy implementable, $\eta^*_N$, is non-increasing in $N$.
Since the $N$-copy implementability of $T^\eta$ is equivalent to the positivity of $(1-\eta) \JL_N + \eta \Id / d$,
$\lambda_{\min} (\JL_N) = -\eta^*_N / d(1-\eta^*_N) $ holds.
Hence $\lambda_{\min} (\JL_N)$ is non-decreasing in $N$.
Therefore, for $N \leq d-2$ cases, it is enough to show that the minimum eigenvalue of $\JL_{d-1}$ is $-1$.
We will only show for $N \geq d-1$ cases.
We give an explicit construction of (unnormalized) eigenvector $\ket{\Psi_N}$ 
which has the eigenvalues shown in Eq.~\eqref{eqap:min_eigen_optimal_transpose}. 
Let $\ket{A_d}$ be the $d$-dimensional totally anti-symmetric state
\begin{align}
\ket{A_d}_{i_0 i_1\cdots i_{d-1}} := \sum_{\sigma \in \mathcal{S}_d} 
\frac{\mathrm{sgn}(\sigma)}{\sqrt{d!}} \ket{\sigma_0}_{i_0} \ket{\sigma_1}_{i_1} \cdots \ket{\sigma_{d-1}}_{i_{d-1}},
\end{align}
and $k_1, \ldots , k_{d-1}$ be integers from $1$ to $N$ which are used to denote the Hilbert spaces.
The eigenvector $\ket{\Psi_N}$ is defined as
\begin{align}
\ket{\Psi_N} := \sum_{k_1 < k_2 < \cdots < k_{d-1}} c(\{k_i\}) \ket{0}^{\otimes N-d+1}_{R} \ket{A_d}_{0 k_1 k_2 \cdots k_{d-1}},
\end{align}
where the coefficients $c(\{k_i\})$ are defined as 
\begin{align}
c(\{k_i\}) :=
\begin{cases}
1 & d\mbox{ is even} \\
\sum_{i=1}^{d-1} (-1)^i k_i & d\mbox{ is odd}
\end{cases},
\end{align}
and the state $\ket{0}^{\otimes N-d+1}_{R}$ is an arbitrary $N-d+1$ qudit state with subscript $R$ denoting that 
this state is in the rest of Hilbert spaces other than $\Hil_0, \Hil_{k_1} ,\ldots, \Hil_{k_{d-1}}$.

Now $\JL_N \ket{\Psi_N}$ can be grouped into two groups as
\begin{align}
&\JL_N \ket{\Psi_N} \notag \\
=& \frac{1}{N} \sum_{k_1 < \cdots < k_{d-1}} \sum_{i \in \{k_i \}} S_{0i} c(\{k_i\}) \ket{0}^{\otimes N-d+1}_{R} \ket{A_d}_{0 k_1 k_2 \cdots k_{d-1}} \notag\\
+& \frac{1}{N} \sum_{k_1 < \cdots < k_{d-1}} \sum_{i \notin \{k_i \}} S_{0i} c(\{k_i\}) \ket{0}^{\otimes N-d+1}_{R} \ket{A_d}_{0 k_1 k_2 \cdots k_{d-1}},\label{eqap:L_N_psi_N}
\end{align}
 depending on $i \in \{k_i \}$ or not.
Since
\begin{align}
&S_{0i} \ket{0}^{\otimes N-d+1}_{R} \ket{A_d}_{0 k_1 k_2 \cdots k_{d-1}} \notag \\
=& 
\begin{cases}
- \ket{0}^{\otimes N-d+1}_{R} \ket{A_d}_{0 k_1 k_2 \cdots k_{d-1}} & i \in \{k_i \} \\
\ket{0}^{\otimes N-d+1}_{R} \ket{A_d}_{i k_1 k_2 \cdots k_{d-1}}  & i \notin \{k_i \}
\end{cases},
\end{align}
the first term of Eq.~\eqref{eqap:L_N_psi_N} becomes
\begin{align}
&\frac{1}{N} \sum_{k_1 < \cdots < k_{d-1}} \sum_{i \in \{k_i \}} S_{0i} c(\{k_i\}) \ket{0}^{\otimes N-d+1}_{R} \ket{A_d}_{0 k_1 k_2 \cdots k_{d-1}} \notag\\
=& - \frac{1}{N} \sum_{k_1 < \cdots < k_{d-1}} \sum_{i \in \{k_i \}} c(\{k_i\}) \ket{0}^{\otimes N-d+1}_{R} \ket{A_d}_{0 k_1 k_2 \cdots k_{d-1}} \notag\\
=& - \frac{d-1}{N} \sum_{k_1 < \cdots < k_{d-1}} c(\{k_i\}) \ket{0}^{\otimes N-d+1}_{R} \ket{A_d}_{0 k_1 k_2 \cdots k_{d-1}} \notag\\
=& - \frac{d-1}{N} \ket{\Psi_N},
\end{align}
where the second equation holds as there are $d-1$ elements in $\{ k_i \}$.
The remaining is to show that the second term of Eq.~\eqref{eqap:L_N_psi_N} vanishes. Here
\begin{align}
&\frac{1}{N} \sum_{k_1 < \cdots < k_{d-1}} \sum_{i \notin \{k_i \}} S_{0i} c(\{k_i\}) \ket{0}^{\otimes N-d+1}_{R} \ket{A_d}_{0 k_1 k_2 \cdots k_{d-1}} \notag\\
=& \frac{1}{N} \sum_{k_1 < \cdots < k_{d-1}} \sum_{i \notin \{k_i \}} c(\{k_i\}) \ket{0}^{\otimes N-d+1}_{R} \ket{A_d}_{i k_1 k_2 \cdots k_{d-1}}, \label{eqap:off_diagonal_iK}
\end{align}
and the summation is taken over the whole set of $\{ i, k_1, k_2 ,\ldots , k_{d-1} \}$ 
with conditions $k_1 < \cdots < k_{d-1}$ and $i \neq k_1, \ldots, k_{d-1}$.
Now consider another set $\{ l_0, l_1, \ldots, l_{d-1} \}$, 
if we sum over the whole set with the condition $l_0 < l_1 < \cdots < l_{d-1}$,
the summation does not match the original one, as in the original summation, 
the only requirement for $i$ is that $i$ is different from the others.
Thus, we add an additional parameter $ 0 \leq j \leq d-1$, 
which is used to specify the number that is not ordered as $i$ in the original summation.
For a fixed parameters $\{ i, k_1, k_2 ,\ldots , k_{d-1} \}$, the corresponding parameters $\{ l_0, l_1, \ldots, l_{d-1}, j \}$
are constructed as follows.
The numbers $l_0, l_1, \ldots, l_{d-1}$ are equal to the numbers $ i, k_1, k_2 ,\ldots , k_{d-1}$ rearranged in an increasing order,
and $j$ is defined such that $i$ is the $(j+1)$-th smallest number in $\{ i, k_1, k_2 ,\ldots , k_{d-1} \}$.
Then we can rewrite Eq.~\eqref{eqap:off_diagonal_iK} as
\begin{align}
\sum_{l_0 < l_1 < \cdots < l_{d-1}} \sum_{j=0}^{d-1} c( \{ l_i \}, j) (-1)^j \ket{0}^{\otimes N-d+1}_{R} \ket{A_d}_{ l_0 l_1 \cdots l_{d-1} }
\end{align}
where the coefficient $(-1)^j$ appears because of the property
$\ket{A_d}_{l_j l_0\cdots l_{j-1} l_{j+1} \cdots l_{d-1}} = (-1)^j \ket{A_d}_{l_0\cdots l_{j-1} l_j l_{j+1} \cdots l_{d-1}}$,
and the coefficient $c( \{ l_i \}, j)$ is equal to $c(\{k_i\})$, which can be written as 
\begin{align}
c( \{ l_i \}, j) = 
\begin{cases}
1 & d\mbox{ is even} \\
\sum_{i < j} (-1)^i l_i + \sum_{i > j} (-1)^{i+1} l_i & d\mbox{ is odd}
\end{cases}. \notag
\end{align}
We show that $\sum_{j=0}^{d-1} c( \{ l_i \}, j) (-1)^j = 0$ as follows.
If $d$ is even, $\sum_{j=0}^{d-1} (-1)^j = 0$ is trivial,
since the summation of $(-1)^j$ over an even number of consecutive integers $j$ is 0.
If $d$ is odd,
\begin{align}
&\sum_{j=0}^{d-1} c( \{ l_i \}, j) (-1)^j \\
&=\sum_{j=0}^{d-1} (-1)^j \left( \sum_{i<j} (-1)^i l_i + \sum_{i>j} (-1)^{i+1} l_i \right)  \\
&= \sum_{i=0}^{d-1} \left( \sum_{j > i} (-1)^{i+j} l_i + \sum_{j < i} (-1)^{i+j+1} l_i \right) \\
&= \sum_{i=0}^{d-1} l_i \left( \sum_{j > i} (-1)^{i+j} - \sum_{j<i} (-1)^{i+j} \right),
\end{align}
where the order of summation is changed in the second equation and we take the summation on $j$ first.
For each $i$, two summations ($j > i$ and $j < i$) consist of an even number of terms in total because $d$ is odd and we take summation only for $ j \neq i$.
If both summations contain an even number of terms, both of them vanish, and the total is 0.
If both summations contain an odd number of terms, 
both of them are equal to $1$ or $-1$, and the total is still 0.
Thus, in both cases, the sum vanishes.

Summing up these results, we obtain $\JL_N \ket{\Psi_N} = - \frac{d-1}{N} \ket{\Psi_N}$.
This is an upper bound for the minimum eigenvalue of $\JL_N$.

\end{document}